\documentclass[11pt]{article}

\usepackage[T1]{fontenc}
\usepackage[english]{babel}
\usepackage{color}
\usepackage{eurosym}
\usepackage{multirow}
\usepackage{listings}

\usepackage[top=20mm, bottom=20mm, left=20mm, right=20mm]{geometry}

\usepackage{graphicx}
\usepackage{float}
\usepackage{rotating}
\usepackage[subfigure]{graphfig}

\usepackage{amsmath,amsthm,amssymb}
\usepackage{bbm}
\usepackage{mathtools}

\usepackage[math]{blindtext}

\usepackage{setspace}
\usepackage{natbib}
\usepackage{hyperref} 

\newtheorem{theorem}{Theorem}[section]

\newtheorem{proposition}[theorem]{Proposition}

\newtheorem{remark}{Remark}
\newtheorem{algorithm}{Algorithm}
\newtheorem{example}{Example}
\newtheorem{definition}{Definition}

\newcommand{\ud}{\mathrm{d}}
\newcommand{\K}{\mathrm{K}}
\newcommand{\T}{\mathcal{T}}
\newcommand{\M}{\mathcal{M}}

\begin{document}
\newgeometry{margin=2cm,top=2cm,bottom=2cm}

\title{
        A structural approach to default modelling with pure jump processes 
        }

\author{Jean-Philippe Aguilar        \thanks{Corresponding author.
        Cov\'ea Finance, Quantitative Research Team, 8-12 rue Boissy d'Anglas, FR-75008 Paris. 
        \newline
        \hspace*{0.5cm} Email: jean-philippe.aguilar@covea-finance.fr, nicolas.pesci@covea-finance.fr, victor.james@covea-finance.fr
        \newline
        \hspace*{0.5cm} The views expressed in this article are those of the authors. They do not purport to reflect the views of Cov\'{e}a Finance.}
        \and
        Nicolas Pesci
        \and
        Victor James
        }
\date{February 11, 2021, revised July 1st, 2021.}

\maketitle
\thispagestyle{empty}

\begin{abstract}
We present a general framework for the estimation of corporate default based on a firm's capital structure, when its assets are assumed to follow a pure jump L\'evy processes; this setup provides a natural extension to usual default metrics defined in diffusion (log-normal) models, and allows to capture extreme market events such as sudden drops in asset prices, which are closely linked to default occurrence. Within this framework, we introduce several pure jump processes featuring negative jumps only and derive practical closed formulas for equity prices, which enable us to use a moment-based algorithm to calibrate the parameters from real market data and to estimate the associated default metrics. A notable feature of these models is the redistribution of credit risk towards shorter maturity: this constitutes an interesting improvement to diffusion models, which are known to underestimate short term default probabilities. We also provide extensions to a model featuring both positive and negative jumps and discuss qualitative and quantitative features of the results. For readers convenience, practical tools for model implementation and GitHub links are also included.

\bigskip


\noindent {\bfseries AMS subject classifications (MSC 2020):} 60E07, 60E10, 62P05, 91G30, 91G40.

\end{abstract}

\newpage
\setcounter{page}{1}

\section{Introduction}\label{sec:intro}

We start the paper by providing a general introduction to the so-called structural approach to credit risk, recalling the main limitations of log-normal models and introducing the more realistic class of pure jump models. We also detail the main contributions of the paper in this context, as well as its overall structure.

\subsection{Structural modelling, credit spread puzzle and jump risk}
A popular approach to credit risk and corporate default modelling is to assume that the value of a firm's assets is driven by a certain stochastic dynamics, and that a default occurs when the realization of this process is lower than the facial value of the firm's debt at its maturity. This approach, which has been called the \textit{structural approach} by \cite{Duffie99}, was pioneered by the works of \cite{Black73} and \cite{Merton74}, the authors introducing a diffusion process to model the underlying asset log returns; this setup is now universally known as the Merton model, and has subsequently been extended to take into account various features, such as the existence of multiple maturities for the firm's debt \citep{Geske77} or the possibility for the default to occur prior to the debt's maturity \citep{Black76,Longstaff95,Leland96}; these latter models are often referred to as {\it first passage} or {\it barrier} models. Extensions to sovereign issuers have also been considered \citep{Gray07}. For a complete overview of default modelling and related topics within the structural approach (as well as in the alternative class of {\it reduced-form} models), we refer to \cite{Duffie12,Lipton13}.

A major difficulty in the structural approach lies in its calibration, because most of the corporate debt is not traded and, therefore, it is not possible to obtain directly the value of a firm's assets by summing its equity and debt values. In this context, assuming that assets log returns are driven by a diffusion process (i.e., a Wiener process, or Brownian motion) is particularly useful, because it allows for a closed-form relationship between a firm's value and its equity, via the celebrated Black-Scholes formula. This closed formula opens the way to simple historical calibrations for the asset's volatility, by solving a system of nonlinear equations with classical solvers such as the Newton-Raphson method in the algorithm developed by \cite{Vassalou04}; other calibrations techniques for diffusion models have also been introduced, for instance based on the book value of the debt in \cite{Eom04} or on maximum likelihood estimations (see details in \cite{Duan04}). 

The Brownian hypothesis, however, has been severely criticized for being unrealistic and, notably, for producing an almost zero default probability for short maturities: if this was true, then short term bonds should have zero credit spread, which is typically not the case, as first noted by \cite{Jones84} (see also many subsequent discussions such as \cite{Lyden00} or \cite{Demchuk09} for instance). This underestimation of short term default probabilities and of theoretical credit spreads is known as the {\it credit spread puzzle} and has been evidenced many times and for every ranking, from high yield to investment grade issuers (see a recent overview in \cite{Huang19}) and even sovereign issuers (see \cite{Duyvesteyn15} and references therein). Among the plausible explanations for this discrepancy are the fact that other factors can influence credit spreads, such as taxes, liquidity premia \citep{D'Amato03} and jump risk \citep{Bai20}.  

This phenomenon is also evidenced when examining 1 year historical default rates; in 2019, the S\&P global default rate (all issuers) was of 1.03\%, and, on the non investment grade or speculative sector corresponding to BB ratings and lower, it reached 2.10\% \citep{SP20b}. With the dramatic COVID19 events, the situation has worsened, and the default rate is expected to reach 10\% in 2021 on the non investment grade sector, which now represents 30 \% of all issuers, at an all time high \citep{SP20a}. Such levels are similar and even higher to those following the global financial crisis (GFC) of 2007-2008 (1 year speculative default rate was of 9.94\% in 2009), and cannot be obtained from the classical Merton model, in particular because the normal assumption fails to reproduce extreme market events. A natural idea is therefore to introduce jumps in the assets dynamics, in order to better capture short term defaults occurring after a brutal drop in the value of a company's assets following, for instance, earning announcements, central bank meetings or major political events.

Jump processes can be of two kinds: jump-diffusion processes, or pure jumps processes, and both are conveniently described by the formalism of L\'evy processes. Jump diffusion processes were introduced in structural credit risk modelling by \cite{Zhou97}, by adding a Poisson process whose discrete jumps are normally distributed to the usual diffusion process, in order to materialize sudden changes in the firm's assets value; other distributions for the magnitude of the Poisson jumps have subsequently been considered, such as negative exponential distributions in \cite{Lipton02}, leading to higher short term probabilities and more realistic credit spread curves. Self exciting Hawkes process have also been introduced in \cite{Ma16}, a model for which an analytical formula for the equity value has been recently derived in \cite{Pasricha21}. Pure jump processes allow for an even richer dynamics and, notably, an interpretation in terms of {\it business time} (differing from the {\it operational time}), or the possibility for jumps to occur arbitrarily often on any time interval; such processes were introduced in the construction of credit risk models in the late 2000s and early 2010s, notably in \cite{Madan08} for one-sided processes (i.e., featuring downward jumps only) and from the point of view of first passage models, and in \cite{Fiorani06,Fiorani10, Luciano09} for double-sided processes, with CDS-based calibrations. Extensions to multivariate processes have also been studied \citep{Marfe12}.

In this paper, we would like to demonstrate that pure jump processes are well suited to the structural approach, thus confirming the initial works cited above, but also that practical formulas can be derived for the equity value, thus allowing for a precise issuer-based calibration. We will also show that incorporating realistic features such as non-normality of returns, jumps or asymmetry, allows to better capture the probability of a default occurrence in particular during turbulent times, and provides a better fit to historical default rates notably for speculative issuers. This approach appears to be particularly relevant in the current COVID19 period: as already mentioned, short term defaults are expected to increase, as they did during former major crisis (GFC, dot com crisis \dots). Moreover, while former peaks of default events  lasted typically one or two years before going back to "usual" default rates, it is likely that, this time, the situation may last far longer, given the enormous amounts of debt that have been accumulated and the feeble revenues that most firms were able to collect between lockdown periods.



\subsection{Contributions of the paper}

We will focus on:
\begin{itemize}
    \item[(a)] Extending the structural approach to the case where the stochastic process is no longer a Brownian motion but a more general L\'evy process, defining credit risk metrics in that context and showing how model parameters can be calibrated using observable issuers data;
    \item[(b)] Deriving closed pricing formulas for the equity value of a firm when the L\'evy process is assumed to be spectrally negative, and implementing the calibration algorithm using these formulas;
    \item[(c)] Showing that obtained default probabilities are higher than the Merton default probabilities for short maturities, and comparing real world default probabilities with historical default rates;
    \item[(d)] Using recent pricing formulas in the context of double-sided processes and, like in (b) and (c), discussing calibration and obtained probabilities.
\end{itemize}

\subsection{Structure of the paper}

In section \ref{sec:dynamics}, we recall basic facts on L\'evy processes, we define credit risk metrics (distance to default, default probability) in this context, and we present the calibration algorithm that extends the algorithm of \cite{Vassalou04} to the case of pure jump processes. Then, in section \ref{sec:one_sided}, we introduce two spectrally negative processes which are similar to the Gamma and inverse Gaussian subordinators, but with a L\'evy measure translated to the negative real axis, allowing us to define the NegGamma and NegIG credit risk models; we derive explicit formulas for equity values in these models and, using the calibration algorithm, we determine model parameters, compute default probabilities (both in the risk-neutral and real world cases) and compare them with Merton probabilities as well as with historical default rates. We proceed to a similar discussion in section \ref{sec:symmetric}, but for a process featuring a double-sided L\'evy measure which, for simplicity, is assumed to be symmetric; equity value in this case is based on series expansions whose terms are powers of the distance to default. Section \ref{sec:conclusion} is dedicated to concluding remarks and future works. For reader's convenience, we have also equipped the paper with an appendix, which summarizes the papers notations, provides details on the data set of test issuers used in the study and on parameters stability, and contains a link to the GitHub where the calibration code can be freely accessed.

\section{Pure jump asset dynamics and risk metrics}\label{sec:dynamics}

In this section, we start by recalling some fundamental concepts on L\'evy processes and asset pricing (we refer to the classical references \cite{Bertoin96} and \cite{Schoutens03} for all technical details), that extend the classical log-normal Merton model. Then, we introduce the generalization of distance to default and default probability in this context, as well as the general algorithm that will be used to calibrate the various models in the rest of the paper.

\subsection{Model setup}\label{subsec:model_setup}

\subsubsection{Model formulation}
Following the classical setup of structural credit risk modelling, we consider a company possessing a total asset $V_A(t)$ at time $t\in[0,T]$, financed by an equity $V_E(t)$ and a zero coupon debt of maturity $T$ and face value $K$:
\begin{equation}
    V_A(t) \, = \, V_E(t) \, + \, K
    .
\end{equation}
At $t=T$, two situations can occur:
\begin{itemize}
    \item[-] $V_A(T)\geq K$: the issuer has enough financial resource to pay off for the full amount of its debt $K$, and in that case its equity still has a positive value equal to $V_A(T) - K$;
    \item[-] $V_A(T) < K$: the company is in default, and in that case its equity value falls down to $0$.
\end{itemize}
In short, we see the equity value of a firm as a European call option written on its assets, whose strike price (resp. maturity) equals the face value (resp. the maturity) of the company's debt; following the usual financial notation, we will therefore write that
\begin{equation}\label{V_E}
    V_E(T) \, = \, [ V_A(T) - K]^+
    .
\end{equation}

We deliberately use a European formulation instead of a first passage one, which, at first sight, could appear more precise as default can occur at all $t\leq T$. The reason is that our purpose is to find a good balance between realistic features (presence of brutal market events, materialized by a pure jump dynamics for $V_A(t)$), and computational tractability (allowing for a simple issuer-based calibration procedure, necessitating only observable equity prices). If we had chosen a first passage formulation, then the equity price \eqref{V_E} would have become a barrier option, for which a simple closed formula exists if $V_A(t)$ is log-normal, but not if it follows a purely discontinuous process, depriving us for a simple calibration algorithm based on the inversion of an equity formula. Extending pricing formulas for barrier options to the case of pure jump processes is, actually, an interesting subject in itself and could be made possible at least in the $\alpha$-stable case by using a generalization of the reflection principle (see \cite{Bingham73}) which states that, if $V_A(t)$ follows an $\alpha$-stable process ($\alpha\in[0,2]$), then
\begin{equation}
    \mathbb{P} \left[ \underset{t\in[0,T]}{\mathrm{sup}} V_A(t) \geq K \right] \, = \, 
    \alpha \mathbb{P} \left[ V_A(T) \geq K \right]
    ,
\end{equation}
which degenerates into the usual Brownian principle when $\alpha\rightarrow 2$. But, as stable distributions have infinite moments except for the degenerate case $\alpha=2$, we would be unable, again, to use a simple moment matching procedure to calibrate model parameters. On the contrary, focusing on a European approach but combining this simple approach with some well-chosen pure jump processes (for which we will be able to derive closed formula for equity prices), will allow to combine both the advantages of computational simplicity and realistic description of the assets behaviour. 

As we will see in the next sections, the pure jump L\'{e}vy processes we will introduce to model the dynamics of $V_A(t)$ will mainly be spectrally negative (that is, feature downward jumps only) because it is this skewed behaviour that has the biggest impact on the credit risk of a firm. The main intuition behind this is that, even if we measure default probability only at $t=T$, downward jumps (materializing a sudden drop in assets value) can occur anytime in $[0,T]$; if such a jump occurs not too far from maturity, and if no positive jumps are allowed, then it will not be possible for asset prices to grow up to the debt level before maturity and default will almost surely occur.
For sake of generality, we will also extend the study to the case of a symmetric process. 

Let us mention that more general pure jump processes have already been considered for the purpose of default modelling (for instance the Variance Gamma process, see \cite{Fiorani10}), featuring one or more supplementary degree of freedom (allowing, for instance, to control the asymmetry of tails). But, for such processes, no simple closed formula can be derived for equity prices, and calibration is typically performed using CDS; as they are not traded for every name and maturity, one has to use CDS indices as proxies. On the contrary, the main novelty of our approach is that we are able to derive convenient closed-form equity formulas, which opens the way to a simple calibration algorithm that only necessitates equity prices as input; the L\'{e}vy processes involved are less sophisticated than an asymmetric Variance Gamma process, but remain far more realistic than the usual diffusion approach.

\subsubsection{Dynamics}\label{subsec:dynamics}
We will assume that, given a filtered probability space $(\Omega = \mathbb{R}_+,\mathcal{F},\{ \mathcal{F} \}_{t\geq 0},\mathbb{P})$, the instantaneous variations of the firm's assets $V_A(t)$ can be written down in local form as
\begin{equation}\label{SDE}
    \frac{\ud V_A(t)}{V_A(t-)} \, = \, r \ud t \, + \, \ud X_t \, , \hspace*{0.5cm} t\in[0,T]
    ,
\end{equation}
under an admissible equivalent martingale measure (EMM), or risk-neutral (RN) measure, traditionally denoted by $\mathbb{Q}$ (see discussion thereafter in subsection \ref{subsec:dd_radp}). In \eqref{SDE}, $r$ is the (continuously compounded) risk-free interest rate and $X_t$ is a L\'evy process, that is, a c\`{a}dlag process satisfying $X_0=0$ ($\mathbb{P}$-almost surely), and whose increments are independent and stationary. Note that when $X_t = W_t$ where $W_t \sim \mathcal{N}(0,\sigma^2 t)$ is a centered Wiener process, then $V_A(t)$ follows a geometric Brownian motion, and in that case \eqref{SDE} corresponds to the traditional Merton model.

The solution to the stochastic differential equation \eqref{SDE} is the exponential process
\begin{equation}\label{SDE_solution}
    V_A(T) \, = \, V_A e^{ ( r + \omega_X) T + X_T}
\end{equation}
where $V_A:=V_A(0)$ and where $\omega_X$ is a convexity adjustment computed in a way that the discounted asset prices are a $\mathbb{Q}$-martingale, i.e.
\begin{equation}\label{martingale}
    \mathbb{E}^{\mathbb{Q}}
    \left[ V_A(T) | V_A \right]
    \, = \, e^{rT} V_A
    .
\end{equation}
As $X_t$ is a L\'evy process, its probability distribution is infinitely divisible and, consequently, its characteristic function admits a semigroup structure. In other words, there exists a function $\psi_X (u)$ called characteristic exponent, or L\'evy symbol, such that $\mathbb{E}\left[ e^{i u X_t} \right] \, = \, e^{t\psi_X(u)}$, and which is entirely characterized by the L\'evy-Khintchine formula
\begin{equation}\label{Levy_Khintchine}
    \psi(u) \, = \, a~i u \, - \, \frac{1}{2} b^2 u^2 \, + \, \int\limits_{ - \infty}^{+\infty} \left( e^{i u x} - 1 - i u x \mathbbm{1}_{\{ |x| < 1 \}}  \right) \Pi_X (\ud x)
    ,
\end{equation}
where $a$ is the drift and is $b$ the Brownian (or diffusion) component. The~measure $\Pi_X(\ud x)$, assumed to be concentrated on $\mathbb{R}\backslash\{0\}$ and satisfy
\begin{equation}
    \int\limits_{-\infty}^{+\infty} \, \mathrm{min} (1,x^2) \, \Pi_X(\ud x) \, < \, \infty
    ,
\end{equation}
is called the L\'evy measure of the process, and~determines its tail behaviour and the distribution of jumps. If $b\neq 0$ and the L\'evy measure has finite total mass, one speaks of a jump-diffusion process and in that case only a finite number of jumps can occur on any finite time interval; among the most popular distributions for the random size of jumps are Gaussian distributions \citep{Zhou97} or exponential distributions \citep{Lipton02}. Non-random discrete (Dirac) jumps have also been introduced in \cite{Lipton09}, and shown to provide remarkably good results when implemented in a model featuring multi-name joint dynamics, notably because they allow to achieve higher levels of correlations than exponential jumps. Among other jump-diffusion models, let us also cite the works by \cite{Hilberink02} or \cite{LeCourtois06} (which is the credit risk counterpart of the geometric Kou model \citep{Kou02} for option pricing). When $b=0$, one speaks of a pure jump process and, if furthermore $\Pi_X(\mathbb{R}) = \infty$, we say that the process has infinite activity (or intensity). For such processes, an infinite number of jumps can occur on any time interval, making for a very rich and complex dynamics; among the most prominent pure jump L\'evy processes are the so-called subordinators (i.e. models for which $\Pi_X(\mathbb{R}_-)=0$) such as the Gamma or Inverse Gaussian processes, or models featuring both positive and negative jumps such as the Variance Gamma (VG) process \citep{Madan98} or the Normal Inverse Gaussian (NIG) process \citep{Barndorff95,Barndorff97}. Let us also mention that the drift $a$ actually depends on the choice of the so-called truncation function (in \eqref{Levy_Khintchine} we choose $i u x \mathbbm{1}_{\{ |x| < 1 \}} )$, but the choice is not unique), and therefore can be eliminated by a suitable choice of truncation. Last, let us remark that, as a direct consequence of the definition of the L\'evy symbol, the condition \eqref{martingale} can be reformulated as:
\begin{equation}\label{martingale_Characteristics}
    \omega_X \, = \, -\psi_X(-i)
    .
\end{equation}
\begin{remark}
\label{remark:omega_BSM}
    If $X_t = W_t$ where $W_t \sim \mathcal{N}(0,\sigma^2 t)$ is a centered Wiener process, it is well known that $\mathbb{E}\left[ e^{i u W_t} \right] \, = \, e^{-\frac{\sigma^2u^2 }{2}t}$. From the point of view of the representation \eqref{Levy_Khintchine}, this means that the L\'evy measure is identically null and that the L\'evy symbol resumes to its diffusion component i.e. $\psi_{W}(u) = -\frac{\sigma^2 u^2}{2}$; it results that the convexity adjustment in this case is
    \begin{equation}
        \omega_{W} \, = \, - \psi_{W}(-i) \, = \, -\frac{\sigma^2}{2}
    \end{equation}
which is the usual Gaussian (Black-Scholes-Merton) adjustment.
\end{remark}

Let us now introduce the quantity \begin{equation}\label{DD}
    k_X \, := \, \log\frac{V_A}{K} + (r+\omega_X)T
\end{equation} 
which, in the context of option pricing, is generally called the log forward moneyness, and let us denote by $f_X(x,t)$ the density of the process $X_t$. Then, we can formulate a proposition:
\begin{proposition}[Equity value]
\label{prop:equity}
Under the dynamics \eqref{SDE}, the equity value of a firm can be written as
\begin{equation}\label{equity}
    V_E \, = K e^{-rT} \, 
    \int\limits_{-k_X}^{\infty} \, 
    \left( e^{k_X+x} - 1 \right) \,
    f_X(x,T) \, 
    \ud x
    .
\end{equation}
\end{proposition}
\begin{proof}
From \eqref{V_E} and the standard theory of martingale pricing, the value of a firm's asset is
\begin{equation}\label{RN_pricing_1}
    V_E 
    \, = \,
    e^{-rT} \, \mathbb{E}^\mathbb{Q} \left[ [V_A(T) - K]^+ | V_A \right]
    .
\end{equation}
Using \eqref{SDE_solution} and the notation \eqref{DD}, we can re-write \eqref{RN_pricing_1} as
\begin{equation}\label{RN_pricing_2}
    V_E 
    \,=\,
    e^{-rT} \mathbb{E}^\mathbb{Q} \left[ [V_A e^{ ( r + \omega_X) T + X_T} - K]^+ | V_A \right]
    \, = \,
    K e^{-rT} \, \mathbb{E}^\mathbb{Q} \left[ [e^{k_X+X_T} - 1]^+ | V_A \right]
    .
\end{equation}
The last expectation can be carried out by integrating all possible realizations for the terminal payoff over the process density, resulting in formula \eqref{equity}.
\end{proof}
\begin{remark}
In the standard option pricing theory, the risk-neutral (martingale) representation formula \eqref{RN_pricing_1} holds true only if the firm's assets (materialized by the process $V_A(t)$) are tradeable, which, in general, is not the case. According to \cite{Ericsson02}, the approach is still licit because it is possible to constitute a portfolio of traded assets that possesses exactly the same dynamics than $V_A(t)$; in other words, even if $V_A(t)$ cannot be traded, it can be replicated by a tradeable portfolio, and, therefore, the models can still be based on the risk-neutral representation (which, in the case of a L\'{e}vy process, is not necessarily unique, see further discussion in section \ref{subsec:dd_radp}).
\end{remark}

\subsection{Risk-neutral distance to default and default probabilities}\label{subsec:dd_rndp}

The occurrence of a default at time $T$ corresponds to the situation where the value of the firm's asset is lower than the facial value of the debt; in this situation, the firm's equity becomes equal to 0 and, therefore, it follows from eq. \eqref{RN_pricing_2} that the probability of such an event is
\begin{align}
    \mathbb{Q}
    \left[ V_A(T) < K | V_A \right] \, & = \, \mathbb{Q} [X_T < -k_X] \noindent \\
    & = \, F_X[-k_X]
\end{align}
where $F_X(.)$ is the cumulative distribution function of the process $X_t$ evaluated at $t=T$:
\begin{equation}
    F_X(x) \, = \, \int\limits_{-\infty}^{x} \, f_X(y,T) \, \ud y
    .
\end{equation}
These observations motivate the following definitions:
\begin{definition}[Risk-neutral Default Metrics]
\label{def:metrics}
    \hspace*{0cm}
    \begin{itemize}
    \item[1.] Under the pure jump dynamics \eqref{SDE}, the risk-neutral distance to default of a firm is defined to be
    \begin{equation}
         k_X \, := \, \log\frac{V_A}{K} + (r+\omega_X)T
         ;
    \end{equation}
    \item[2.]The corresponding risk-neutral default probability of the firm is defined to be
        \begin{equation}
         P_X \, := \, F_X(-k_X)
         .
    \end{equation}
    \end{itemize}
\end{definition}

\begin{example}[Merton model]
    If $X_t = W_t$ where $W_t\sim\mathcal{N}(0,\sigma^2 t)$, then the process admits the following density (heat kernel):
    \begin{equation}
        f_W(x,t) \, = \, \frac{1}{\sigma\sqrt{2\pi t}} e^{-\frac{x^2}{2\sigma^2 t}}
        .
    \end{equation}
    The cumulative distribution function is therefore
    \begin{equation}
        F_W(x) \, = \, \int\limits_{-\infty}^x \, 
        \frac{1}{\sigma\sqrt{2\pi t}}     e^{-\frac{y^2}{2\sigma^2 t}} \, \ud y
        \, = \, N \left( \frac{x}{\sigma\sqrt{t}} \right)
        \end{equation}
    where $N(.)$ is the cumulative distribution function of the standard normal distribution $\mathcal{N}(0,1)$. It follows from remark \ref{remark:omega_BSM} and definition \ref{def:metrics} that the risk-neutral default probability is 
    \begin{align}
        F_W \left( -\left( \log\frac{V_A}{K} + (r-\frac{\sigma^2}{2})T  \right) \right)
        \, & = \, N\left( - \frac{\log\frac{V_A}{K} + (r-\frac{\sigma^2}{2})T}{\sigma\sqrt{T}} \right) \nonumber \\
        & \, := \, N(-d_2)
    \end{align}
    where $d_2$ (following the usual Black-Scholes notation) is often interpreted as the risk-neutral distance to default in the Merton model.
\end{example}

\subsection{Actual (real-world) distance to default and default probabilities}\label{subsec:dd_radp}

The metrics introduced in definition \ref{def:metrics} are risk-neutral metrics, because in equation \eqref{SDE} we have chosen to specify the dynamics for $V_A(t)$ directly under the risk-neutral, or martingale measure $\mathbb{Q}$. To be fully rigorous, let us mention that $\mathbb{Q}$ is unique only if the market is complete and that, when L\'evy processes are involved, the market is in general incomplete; this means that several equivalent martingale measures exist, and that the choice \eqref{martingale_Characteristics} is not unique - this is however the standard choice (see details e.g. in \cite{Schoutens03}). We may note that, using the notations defined in appendix \ref{app:notations}, equation \eqref{martingale_Characteristics} can be re-formulated in terms of the cumulant generating function as
\begin{equation}
    \omega_X \, = \, -\kappa_X(1)  
    ,
\end{equation}
allowing to define our choice for the change of measure via the following Radon-Nikodym derivative
\begin{equation}
    \frac{\ud \mathbb{Q}_{\mathcal{F}_t}}{\ud \mathbb{P}_{\mathcal{F}_t}}  
    \, = \, 
    \frac{e^{X_t}}{M_X(1,t)}
    \, = \,
    e^{X_t + \omega_X t}
    .
\end{equation}
Under this definition, the $\mathbb{Q}$ measure is sometimes referred to as the {\it exponentially tilted} measure, or the {\it Esscher transform} of the physical measure.

Risk-neutral default probabilities can be used as a convenient proxy for the estimation of actual (also called physical, or real-world) default probabilities, as they essentially convey the same information and sensitivities to the variables such as debt or model parameters (see a discussion e.g. in \cite{Delianedis03}); moreover, they are well-suited to derivatives-based calibrations (such as CDS quotes), and it is common market practice that xVA calculations are commonly made under the $\mathbb{Q}$ measure \citep{Gregory15}. However, the risk-free rate $r$ in the dynamics \eqref{SDE} generally differs from the actual drift $\overline{r}$ at which the assets grow\footnote{The actual drift is often denoted by $\mu$ in the literature; we do not use this notation, in order to avoid confusion with the mean parameter for the inverse Gaussian subordinator and the standardized central moments.}, and computing actual default probabilities can be necessary if one wishes to perform precise default computations, in particular in the framework of portfolio exposure (see for instance a discussion in \cite{Stein16}); in this context, instead of specifying the dynamics \ref{SDE} under the $\mathbb{Q}$-measure, we can choose to specify it under the real-world measure $\mathbb{P}$ and, doing so, to replace the risk-free rate $r$ is by the actual drift $\overline{r}$; the analogue to the risk-neutral metrics \ref{def:metrics} are the actual (or real-world) risk metrics
\begin{equation}\label{def:actual_metrics}
    \overline{k}_X \:= \, \log\frac{V_A}{K} + (\overline{r}+\omega_X) T \, ,
    \quad
    \overline{P}_X \, := \, F_X(-\overline{k}_X)
\end{equation}
from which it follows that risk-neutral and actual default probabilities coincide only if $r=\overline{r}$ (that is, if the underlying assets grow at exactly the risk-free rate). In subsection \ref{subsec:calibration}, we will explain how both model parameters and actual drift can be calibrated using historical data and a moment-matching algorithm.

\subsection{Historical calibration algorithm}\label{subsec:calibration}

\subsubsection{Model parameters}

Given a period of $N$ consecutive trading days $t_1,t_2,\dots,t_N$, with, for simplicity, constant spacing $\Delta:=t_{n}-t_{n-1}$ for all $n\in 1\dots  N$, we constitute a time-series of $N-1$ realizations of the log-returns of $V_A$; using \eqref{SDE_solution}, it can be written as
\begin{equation}\label{T_x_risk_neutral}
    \mathcal{T}_X \, := \, 
    \left\{
        \log\frac{V_A(t_{n+1})}{V_A(t_n)} \, , \, n=1 , \dots, N-1  
    \right\}
    \, = \, 
    \left\{
    (r + \omega_X) \Delta + X_\Delta \, \, ,  n=1 \dots N-1
    \right\}
    .
\end{equation}
In our applications, we will choose a constant spacing parameter $\Delta=1$, i.e., daily data, but other choices of data frequency are of course possible. We denote by $\mu_X^+$ the set of moments of order greater than 2 of $\T_X$ (see definitions \eqref{def_moments} and \eqref{def_moments_set} in appendix \ref{app:notations}). We assume that the model has $N_X$ parameters (or equivalently, that the process $X_t$ has $N_X$ degrees of freedom) grouped in a set denoted by $a_X$, and we choose a subset $\mu_{N_X} \subset \mu_X^+$ of $N_X$ moments of $\T_X$ of order greater than 2; we assume that the moments of $\T_X$ are known in terms of the model parameters, that is, that we can write $\mu_{N_X} = \M_X(a_X)$ for some bijective function $\M_X$. Last, we assume that, using proposition \ref{prop:equity}, we were able to derive a closed pricing formula for the equity price, that is, that we have found a locally invertible function $P_{a_X}$ such that
\begin{equation}
    V_E(t) \, = \, P_{a_X}(V_A(t))
    \quad
    \forall t\in[0,T]
.
\end{equation}
Note that we focus on the dependence on the model parameters $a_X$ but of course the $P_{a_X}$ function can also depend on other external parameters such as maturity $T$ or debt level $K$ for instance. The purpose of algorithm \ref{algorithm_calibration} is to determine $a_X$ (seen as a vector) as the limit of a sequence of vectors $a_X^{(k)}$ calibrated from observable market data, i.e., formally,
\begin{equation}
    a_X
    \, = \, 
    \lim_{k \rightarrow \infty} a_{X}^{(k)}
    .
\end{equation}

\begin{algorithm}[Calibration]
\label{algorithm_calibration}
    \hspace*{0cm}
    \begin{itemize}
        \item[(1)] Choose $N_X$ real numbers, group them under the form of a vector $a_X^{(0)}$ and, for each date $t_n$, define
        \begin{equation}
            V_A^{(0)}(t_n) 
            \, = \, 
            P_{a_X^{(0)}}^{-1} \left( V_E(t_n) \right)
            .
        \end{equation}
        \item[(2)] Denote $\T_X^{(0)} := \left\{ \log\frac{V_A^{(0)}(t_{n+1})}{V_A^{(0)}(t_{n})} , n=1 \dots N-1     \right\}$ and $\mu_{N_X}^{(0)}$ its associated subset of  $N_X$ moments of order greater than 2, and define
        \begin{equation}
            a_X^{(1)}
            \, = \, 
            \M_X^{-1} \left( \mu_{N_X}^{(0)} \right)
            .
        \end{equation}
        \item[(3)] Repeat step (1) and (2): introduce
        \begin{equation}
            V_A^{(1)}(t_n) 
            \, = \, 
            P_{a_X^{(1)}}^{-1} \left( V_E(t_n) \right)
            , 
            \quad
            \T_X^{(1)} := \left\{ \log\frac{V_A^{(1)}(t_{n+1})}{V_A^{(1)}(t_{n})} , n=1 \dots N-1  \right\}
        \end{equation}
        and $\mu_{N_X}^{(1)}$ the subset of $N_X$ moments of order greater than 2 of $\T_X^{(1)}$, and define
        \begin{equation}
            a_X^{(2)}
            \, = \, 
            \M_X^{-1} \left( \mu_{N_X}^{(1)} \right)
            .
        \end{equation}
        \item[(4)] Stop the algorithm once a desired level of precision is attained, e.g. $\vert\vert a_X^{(k)} - a_X^{(k-1)} \vert\vert_{\infty} < 10^{-3}$.
\end{itemize}
\end{algorithm}

\subsubsection{Actual drift}\label{subsec:actual_drift}
When it comes to estimating actual (real world) default probabilities, algorithm \ref{algorithm_calibration} remains valid but we have, furthermore, to estimate the actual drift $\overline{r}$. We note that, under the $\mathbb{P}$ measure, the series of log returns \eqref{T_x_risk_neutral} becomes
\begin{equation}\label{T_x_actual}
    \overline{\mathcal{T}}_X \, := \, 
    \left\{
        \log\frac{V_A(t_{n+1})}{V_A(t_n)} \, , \, n=1 , \dots, N-1  
    \right\}
    \, = \, 
    \left\{
    (\overline{r} + \omega_X) \Delta + X_\Delta \, \, ,  n=1 \dots N-1
    \right\}
\end{equation}
and, therefore, the associated first order moment (the mean) reads
\begin{equation}
    \overline{\mu}_X^{(1)}
    \, = \, 
    \mathbb{E} \left[ \overline{\T}_X \ \right]
    \, = \, 
    (\overline{r} + \omega_X )\Delta + \mathbb{E} \left[ X_\Delta \right]
\end{equation}
where the convexity adjustment $\omega_X$ depends on the model parameters that have been calibrated in algorithm \ref{algorithm_calibration}. It follows immediately that
\begin{equation}\label{actual_drift}
    \overline{r} 
    \, = \, 
    \frac{\overline{\mu}_X^{(1)} - \mathbb{E} \left[ X_\Delta \right]}{\Delta} - \omega_X
    .
\end{equation}
\begin{remark}
    In the Merton case, choosing $\Delta=1$, we have $\mathbb{E}[W_1]=0$ and $\omega_W = -\frac{\sigma^2}{2}$ and therefore
    \begin{equation}
        \overline{r} \, = \, \overline{\mu}_W^{(1)} + \frac{\sigma^2}{2}
        .
    \end{equation}
\end{remark}

\section{One-sided models}\label{sec:one_sided}

In this section, we will be particularly interested in the case where the L\'evy process $X_t$ in \eqref{SDE} is spectrally negative, that is, features negative (downward) jumps only; this will lead us to introduce two spectrally negative pure jump processes, namely the NegGamma and NegIG processes. Let us note that one-sided processes have already been applied to default modelling, e.g. in \cite{Madan08}, but from a point of view of negative subordinators, and with a first passage time approach and a CDS-based calibration. Our approach is different: we will first establish closed formulas for the equity value of the firms under these dynamics and show how they can be used for historical issuer-based calibrations and estimation of default probabilities.

\subsection{Negative Gamma (NegGamma) credit risk model}\label{subsec:Neg_Gamma}
We define the NegGamma process to be the pure jump process whose L\'evy measure reads
\begin{equation}
    \Pi_G (\ud x) \, = \, \rho \frac{e^{-\lambda |x|}}{|x|} \, \mathbbm{1}_{\{ x< 0\}} \, \ud x
\end{equation}
for some positive real numbers $\lambda$ and $\rho$; within this parametrization $\lambda$ is called the rate parameter, and $\rho$ the shape parameter. It follows from the L\'evy-Khintchine representation \eqref{Levy_Khintchine} that the L\'evy symbol of the NegGamma process is, with an adapted choice of truncation function,
\begin{equation}\label{psi_Gamma}
    \psi_{G}(u)
    \,  = \, \int\limits_{-\infty}^{+\infty} \, (e^{iux} - 1) \, \Pi_G(\ud x)
    \, = \,
    - \rho \log\left( 1 + i\frac{u}{\lambda}  \right).
\end{equation}
The cumulant generating function is therefore
\begin{equation}
    \kappa_G(p) \, = \, \psi_G(-ip) \, = \,  - \rho \log\left( 1 + \frac{p}{\lambda}  \right)
\end{equation}
from which we easily deduce the cumulants $\kappa_G^{(n)} \, = \, (-1)^n (n-1)! \, \rho / \lambda^n$ (see further details e.g. in \cite{Kuchler08}). In particular, we obtain, for the standardized moments of the log returns ($\Delta=1$): 
\begin{equation}\label{negGamma_cumulants}
    \left\{
    \begin{aligned}
        &  \overline{\mu}_G^{(1)} \,  = \, \overline{r} + \rho\log\left( 1+ \frac{1}{\lambda} \right) -\frac{\rho}{\lambda} \quad \mathrm{(Mean)}  \\
        & \mu_G^{(2)} \,  = \, \frac{\rho}{\lambda^2} \, \mathrm{(Variance)} \\
        & \mu_G^{(3)} \,  = \, -\frac{2}{\sqrt{\rho}} \, \mathrm{(Skewness)} \\
        & \mu_G^{(4)} - 3  = \, \frac{6}{\rho} \, \mathrm{(Excess \, kurtosis)}
        .
    \end{aligned}
    \right.
\end{equation}

\subsubsection{Default metrics}\label{subsubsec:Gamma_Default}

It follows from \eqref{psi_Gamma} that the martingale adjustment for the NegGamma process is 
\begin{equation}\label{omega_gamma}
    \omega_G \, = \, \rho\log\left( 1+ \frac{1}{\lambda} \right)
\end{equation}
and therefore the distance to default is 
\begin{equation}
    k_G \, = \, \log\frac{V_A}{K} + \left(r + \rho\log\left( 1+ \frac{1}{\lambda} \right)\right) T
    .
\end{equation}
The characteristic function also follows immediately from \eqref{psi_Gamma}:
\begin{equation}
    \Psi_G(u,t) \, = \, \left( 1 + i\frac{u}{\lambda}  \right)^{-\rho t}
\end{equation}
and, inverting this Fourier transform, we get the density function
\begin{equation}\label{density_Gamma}
    f_{G}(x,t) \, = \, \frac{\lambda}{\Gamma(\rho t)} \, |\lambda x|^{\rho t -1}e^{-\lambda|x|} \, \mathbbm{1}_{ \{ x<0 \} }
    ,
\end{equation}
from which we obtain the cumulative distribution function:
\begin{equation}
    F_G(x,T) \, = \, 
    \left\{
    \begin{aligned}
         & \frac{\Gamma(\rho T,-\lambda x)}{\Gamma(\rho T)} \hspace*{0.5cm} \mathrm{if} \hspace*{0.2cm} x<0 \\
         & 1 \hspace*{0.5cm} \mathrm{if} \hspace*{0.2cm} x\geq 0
    \end{aligned}
    \right.
\end{equation}
where $\Gamma(a,z)$ stands for the upper incomplete Gamma function (see appendix \ref{app:notations} and \cite{Abramowitz72}). As a consequence of definition \ref{def:metrics}, the default probability of a firm thus writes
\begin{equation}
    F_G(-k_G,T) \, = \, 
    \left\{
    \begin{aligned}
         & \frac{\Gamma(\rho T,\lambda k_G)}{\Gamma(\rho T)} \hspace*{0.5cm} \mathrm{if}\hspace*{0.2cm} k_G>0 \\
         & 1 \hspace*{0.5cm} \mathrm{if} \hspace*{0.2cm} k_G\leq 0
         .
    \end{aligned}
    \right.
\end{equation}
As a consequence of the definition of the upper incomplete Gamma function, it is clear that $F_G(-k_G,T)\rightarrow 1$ when $k_G\rightarrow 0$, and that there is a 100\% default probability as soon as the distance to default is null or negative. In other words, there is no more chance for $V_A(T)$ to be greater than $K$ as soon as $k_G\leq 0$, which is a consequence of the occurrence of downward jumps only.

\subsubsection{Equity value}\label{subsubsec:Gamma_equity}

Let $\gamma(a,z)$ denote the lower incomplete gamma function, such that the equality $\gamma(a,z)+\Gamma(a,z)=\Gamma(a)$ holds for all $z\geq 0$.

\begin{proposition}
\label{prop:GammaNeg_equity_value}
The equity value of a firm in the NegGamma credit risk model can be written as:
\begin{equation}\label{GammaNeg_equity_value}
    V_E \, = \, 
    \left\{
    \begin{aligned}
    & V_A \, \frac{\gamma(\rho T, (\lambda+1)k_G)}{\Gamma(\rho T)}
    \, - \, Ke^{-rT} \, \frac{\gamma(\rho T, \lambda k_G)}{\Gamma(\rho T)} \hspace*{0.5cm} \mathrm{if } \,\, k_G>0
    ,
    \\
    & 0 \hspace*{0.5cm} \mathrm{if} \,\, k_G\leq 0
    .
    \end{aligned}
    \right.
\end{equation}
\end{proposition}
\begin{proof}
If $k_G >0$, then using proposition \ref{prop:equity} with the density \eqref{density_Gamma} and changing the variable $x\rightarrow -x$, we can write the equity value as:
\begin{equation}
    V_E \, = \, 
    \frac{\lambda^{\rho T}}{\Gamma(\rho T)} Ke^{-rT}
    \left[
    e^{k_G} \int\limits_0^{k_G} e^{-(\lambda + 1) x} x^{\rho T - 1} \ud x
    - 
    \int\limits_0^{k_G} e^{-\lambda x} x^{\rho T - 1} \ud x
    \right]
    .
\end{equation}
Using the definition of the lower incomplete gamma function (see \eqref{lower_gamma} in appendix \ref{app:notations}) and recalling that $Ke^{-rT}e^{k_G}=V_A e^{\omega_G T}$, we have:
\begin{equation}
   V_E \, = \, V_A e^{\omega_G T} \left( \frac{\lambda}{\lambda+1} \right)^{\rho T} 
    \frac{\gamma(\rho T, (\lambda+1)k_G)}{\Gamma(\rho T)}
    - Ke^{-rT} \frac{\gamma(\rho T, \lambda k_G)}{\Gamma(\rho T)}
\end{equation}
which simplifies into the first equation of \eqref{GammaNeg_equity_value} because $e^{\omega_G T}=((\lambda+1)/\lambda)^{\rho T}$.

\noindent If $k_G \leq 0$, then $(-k_G , \infty) \subset \mathbb{R}_+$ and therefore the integral \eqref{equity} is equal to zero as the density \eqref{density_Gamma} is supported by the real negative axis only.
\end{proof}

\subsubsection{Parameter calibration}\label{subsubsec:Gamma_calibration}

There are 2 parameters to calibrate ($N_G=2)$, $\lambda$ and $\rho$, and therefore it suffices to use two central moments in algorithm \ref{algorithm_calibration}, for instance the variance $\mu_G^{(2)}$ and excess kurtosis $\mu_G^{(4)}-3$. Given the formulas \eqref{negGamma_cumulants}, we introduce the function
\begin{equation}\label{M_G}
    \mathcal{M}_G \, : 
    \begin{bmatrix}
        \lambda \\
        \rho
    \end{bmatrix}
    \, \longrightarrow \, 
    \begin{bmatrix}
        \rho / \lambda^2 \\
         6 / \rho
    \end{bmatrix}
    .
\end{equation}
Let $\{ V_E\}$ denote the time series of observable prices during a one year period preceeding the calculation date, and let us initiate we initiate algorithm \ref{algorithm_calibration} by choosing $a_G^{(0)}:= (\mathrm{Variance}(\{V_E\},\mathrm{ex.Kurtosis}(\{V_E\}))$. Following algorithm \ref{algorithm_calibration}, we have to evaluate $P_{a_G^{(0)}}^{-1}$ which can be performed via a classical Newton-Raphson algorithm for the pricing formula in proposition \ref{prop:GammaNeg_equity_value}, while the function $\M_G^{-1}$ is  straightforward to obtain from definition \eqref{M_G}. Calibration results are displayed in table \ref{tab:negGamma_calibration} for a data set constituted of investment grade and speculative issuers with various rankings (see details in table \ref{tab:data_set} in appendix \ref{app:data_set}); speed of convergence (typical number of steps needed for algorithm \ref{algorithm_calibration} to converge) is discussed in table \ref{tab:conv_analysis}, and stability of the parameters is discussed in appendix \ref{app:time_stability}.



\begin{table}[ht]
 \caption{Parameter calibration for several issuers in the NegGamma credit risk model: asset value $V_A$ (in MM GBP or EUR), rate $\lambda$ and shape $\rho$. We choose a maturity $T$=1 year and a risk free interest rate $r=0\%$. Historical calibration performed on a 1 year interval (28 October 2019 to 13 October 2020, corresponding to 252 observations).}
 \label{tab:negGamma_calibration}       
 \centering
 \begin{tabular}{lc|ccc|cc}
 \hline
  \multicolumn{2}{c}{} & \multicolumn{3}{c}{NegGamma model} & \multicolumn{2}{c}{Merton model} \\ 
  \hline
 & $K$ (in MM) & $V_A$ (in MM) & $\lambda$ & $\rho$ & $V_A$ (in MM) & $\sigma$ \\
 SAP GY & 16 196 & 180 913 & 3.280 & 0.888 & 180 914 & 28.73\% \\
 MRK GY & 14 180 & 70 763 & 3.224 & 0.645 & 70 766 & 24.87\%  \\ 
 AI FP & 14 730 & 78 928 & 3.194 & 0.559 & 78 931 & 23.40\% \\
 SU FP & 8 473 & 71 471 & 2.200 & 0.510 & 71 474 & 32.45\% \\
 CRH LN & 10 525 & 33 935 & 2.700 & 0.684 & 33 965 & 30.38\% \\
 DAI GY & 161 780 & 213 453 & 6.736 & 0.530 & 214 039 & 10.78\% \\
 VIE FP & 16 996  & 27 243 & 4.102 & 0.452 & 27 319 & 16.14\% \\
 SRG IM & 14 774 & 29 527 & 2.834 & 0.310 & 29 585 & 19.38\% \\
 AMP IM & 1 339 & 8 627 & 1.784 & 0.414 & 8 629 & 35.95\% \\
 FR FP  & 4 879 & 11 379  & 2.746 & 0.774 & 11 415 & 31.98\% \\
 EO FP & 4 838 & 9 993 & 3.786 & 1.129 & 10 023 & 27.75\% \\
 GET FP & 4 498 & 11 658 & 3.230 & 0.612 & 11 675 & 23.98\% \\
 LHA GY & 10 106 & 14 635 & 4.074 & 0.784 & 14 730 & 21.61\% \\
 PIA IM & 609 & 1 491 & 4.138 & 1.050 & 1 492 & 24.54\% \\
 CO FP & 14 308 & 16 445 & 11.896 & 0.745 & 16 494 & 7.11\% \\
  \hline
\end{tabular}
\end{table}



\subsection{Negative Inverse Gaussian (NegIG) credit risk model}\label{subsec:Neg_IG}

We define the NegIG process to be the pure jump process whose L\'evy measure reads
\begin{equation}
    \Pi_I (\ud x) \, = \, 
    \sqrt{\frac{\lambda}{2\pi}}
    \frac{e^{-\frac{\lambda}{2\mu^2}|x|}}{|x|^{\frac{3}{2}}}
    \,
    \mathbbm{1}_{\{x < 0\}} 
    \, \ud x
\end{equation}
for some positive real numbers $\lambda$ (shape) and $\mu$ (mean); from the L\'evy-Khintchine representation \eqref{Levy_Khintchine}, we deduce the L\'evy symbol for the NegIG process:
\begin{equation}\label{psi_IG}
    \psi_{I}(u)
    \,  = \, \int\limits_{-\infty}^{+\infty} \, (e^{iux} - 1) \, \Pi_I(\ud x)
    \, = \,
    \frac{\lambda}{\mu}
    \left(
    1-\sqrt{1+ 2 i u \frac{\mu^2}{\lambda}}
    \right)
    .
\end{equation}
The cumulant generating function is therefore
\begin{equation}
    \kappa_I(p) \, = \, \psi_I(-ip) \, = \,  
    \frac{\lambda}{\mu}
    \left(
    1-\sqrt{1+ 2 p \frac{\mu^2}{\lambda}}
    \right)
\end{equation}
from which we easily deduce the cumulants $\kappa_I^{(n)}$  and the corresponding central moments:
\begin{equation}\label{cumulants_negIG}
    \left\{
    \begin{aligned}
        &  \overline{\mu}_I^{(1)} \,  = \, \overline{r} +  \frac{\lambda}{\mu} 
        \left( \sqrt{1+2\frac{\mu^2}{\lambda}} - 1  
        \right) - \mu \quad \mathrm{(Mean)}  \\
        & \mu_I^{(2)} \,  = \, \frac{\mu^3}{\lambda} \, \mathrm{(Variance)} \\
        & \mu_I^{(3)}  \,  = \, -3 \sqrt{\frac{\mu}{\lambda}} \, \mathrm{(Skewness)} \\
        & \mu_I^{(4)} - 3 \,  = \, \frac{15 \mu}{\lambda} \, \mathrm{(Excess \, kurtosis)}
        .
    \end{aligned}
    \right.
\end{equation}

\subsubsection{Default metrics}\label{subsubsec:IG_Default}

It follows from \eqref{psi_IG} that the martingale adjustment for the NegIG process is 
\begin{equation}\label{omega_IG}
    \omega_I \, = \, \frac{\lambda}{\mu} 
    \left( \sqrt{1+2\frac{\mu^2}{\lambda}} - 1  
    \right)
\end{equation}
and therefore the distance to default is 
\begin{equation}
    k_I \, = \, \log\frac{V_A}{K} + \left(r + \frac{\lambda}{\mu} \left( \sqrt{1+2\frac{\mu^2}{\lambda}} - 1  
    \right)\right) T
    .
\end{equation}
The characteristic function also follows immediately from \eqref{psi_IG}:
\begin{equation}
    \Psi_I(u,t) \, = \, e^{ \frac{\lambda}{\mu}
    \left(
    1-\sqrt{1+ 2 i u \frac{\mu^2}{\lambda}}
    \right)
    t
    }
\end{equation}
and, inverting this Fourier transform, we get the density function
\begin{equation}\label{density_IG}
    f_{I}(x,t) \, = \,
    \sqrt{\frac{\lambda t^2}{2\pi}}
    \frac{e^{ -\lambda\frac{(|x|-\mu t)^2}{2\mu^2|x|}}}{|x|^{\frac{3}{2}}} 
    \, \mathbbm{1}_{ \{ x<0 \} }
    .
\end{equation}
From definition \ref{def:metrics} and Shuster's integrals (see \eqref{Shuster} in appendix \ref{app:notations}), we therefore obtain the default probability in the negIG model:
\begin{equation}
    F_I(-k_I,T) \, = \, 
    \left\{
    \begin{aligned}
         & N \left( - \sqrt{\frac{\lambda T^2}{k_I}} \left( \frac{k_I}{\mu T} -1 \right) \right)
        \, - \, 
        e^{2\frac{\lambda T}{\mu}} \, 
        N \left( - \sqrt{\frac{\lambda T^2}{k_I}} \left( \frac{k_I}{\mu T} + 1 \right) \right) \hspace*{0.5cm} \mathrm{if}\hspace*{0.2cm} k_I>0 \\
         & 1 \hspace*{0.5cm} \mathrm{if} \hspace*{0.2cm} k_I\leq 0
         .
    \end{aligned}
    \right.
\end{equation}
where N(.) stands for the standard normal cumulative distribution function (see apprendix \ref{app:notations}).

\subsubsection{Equity value}\label{subsubsec:IG_equity}
Let us define the function
\begin{equation}
    \varphi (x, t, \lambda, \mu) \, := \, 
    N \left( \sqrt{\frac{\lambda t^2}{x}} \left( \frac{x}{\mu t} -1 \right) \right)
    \, + \, 
    e^{2\frac{\lambda t}{\mu}} \, 
    N \left( - \sqrt{\frac{\lambda t^2}{x}} \left( \frac{x}{\mu t} + 1 \right) \right)
    .
\end{equation}

\begin{proposition}
\label{prop:IGNeg_equity_value}
The equity value of a firm in the NegIG credit risk model can be written as:
\begin{equation}\label{IGNeg_equity_value}
    V_E \, = \, 
    \left\{
    \begin{aligned}
    & V_A \,
    \varphi \left(  k_I \sqrt{1+2\frac{\mu^2}{\lambda}} , T, \lambda \sqrt{1+2\frac{\mu^2}{\lambda}} , \mu \right)
    \, - \, 
    K e^{-rT} 
    \varphi \left( k_I , T, \lambda , \mu \right)
    \hspace*{0.5cm} \mathrm{if } \,\, k_I>0
    ,
    \\
    & 0 \hspace*{0.5cm} \mathrm{if} \,\, k_I\leq 0
    .
    \end{aligned}
    \right.
\end{equation}
\end{proposition}
\begin{proof}
    If $k_I>0$, then using proposition \ref{prop:equity} with the density \eqref{density_IG}, and changing the variable $x\rightarrow -x$, we can write the equity value as:
    \begin{equation}
        V_E \, = \, 
        \sqrt{\frac{\lambda T^2}{2\pi}} 
        Ke^{-rT}
        \left[ 
        e^{k_I} \int\limits_0^{k_I}
        \frac{e^{-x-\lambda\frac{(x-\mu T)^2}{2\mu^2x}}}{x^{3/2}} \, \ud x
        \, - \, 
        \int\limits_0^{k_I}
        \frac{e^{-\lambda\frac{(x-\mu T)^2}{2\mu^2x}}}{x^{3/2}} \, \ud x
        \right]
    \end{equation}
    The second integral can be performed directly from formula \eqref{Shuster} and the definition of the $\varphi(.)$ function; the first integral can be carried out in a similar way, after completing the square
    \begin{equation}
        -x-\lambda\frac{(x-\mu T)^2}{2\mu^2 x } \, = \, 
        - \frac{\lambda 
        \left( \sqrt{1+2 \mu^2 / \lambda }x - \mu T \right)^2}{2\mu^2 x}
        +
        \frac{\lambda T}{\mu}
        \left( 1 - \sqrt{ 1+2\frac{\mu^2}{\lambda} } \right)
    \end{equation}
    and recalling that 
    \begin{equation}
        Ke^{-rT}e^{k_I} 
        \, = \, 
        V_A e^{\omega_I T}
        \, = \, 
        V_A e^{\frac{\lambda T}{\mu} 
        \left( \sqrt{1+2\mu^2 / \lambda} - 1  
        \right)}
        .
    \end{equation}
Last, if $k_I\leq 0$, then $(-k_I , \infty) \subset \mathbb{R_+}$ and therefore the integral \eqref{equity} is equal to zero as the density \eqref{density_IG} is supported by the real negative axis only.
\end{proof}

\subsubsection{Parameter calibration}\label{subsubsec:IG_calibration}
Calibration process for the NegIG parameters $\lambda$ and $\mu$ goes like the NegGamma calibration in sub-subsection \ref{subsubsec:Gamma_calibration}, by using variance and excess kurtosis, and by introducing the function
\begin{equation}\label{M_I}
    \M_I \, : 
    \begin{bmatrix}
        \lambda \\
        \mu
    \end{bmatrix}
    \, \longrightarrow \, 
    \begin{bmatrix}
        \mu^3 / \lambda \\
        15 \mu / \lambda
    \end{bmatrix}
\end{equation}
whose inverse $\M_I^{-1}$ is straightforward to compute. 



\subsection{Comparison of default probabilities and discussion}\label{subsec:onesided_default_probabilities}

\subsubsection{Risk-neutral default probabilities}

In table \ref{tab:negGamma_negNIG_Merton}, we compare the risk-neutral default probabilities obtained for the set of issuers (see details in appendix \ref{app:data_set}) with our NegGamma and NegIG models, to the one obtained via the classical Merton model, for two distinct horizons $T$ (1 year and 5 years).

\begin{table}[ht]
 \caption{Comparison of NegGamma, NegIG and Merton risk-neutral default probabilities, for time horizons T=1 and 5 years and for several categories/ratings.}
 \label{tab:negGamma_negNIG_Merton}       
 \centering
 \begin{tabular}{lcc|cc|cc|cc}
 \hline
  & Rating & Category &  \multicolumn{2}{c|}{NegGamma} & 
  \multicolumn{2}{c|}{NegIG} & 
  \multicolumn{2}{c}{Merton}
  \\ 
  \hline
  & &  & $T$= 1 y. & $T$=5 y. & $T$= 1 y. & $T$=5 y. &  $T$=1 y. & $T$=5 y. \\
  SAP GY & A & Inv & 0.01\% & 0.47\% & 0.01\% & 0.46\% & 0.00\% & 0.03\% \\
  MRK GY & A & Inv & 0.12\% & 1.90\% & 0.12\% & 1.84\% & 0.00\% & 0.47\% \\
  AI FP & A- & Inv & 0.08\% & 1.26\% & 0.08\% & 1.22\% & 0.00\% & 0.16\% \\
  SU FP & A- & Inv & 0.14\% & 2.02\% & 0.14\% & 1.95\% & 0.00\% & 0.51\% \\
  CRH LN & BBB+ & Inv & 1.10\% & 10.11\% & 1.02\% & 9.97\% & 0.01\% & 10.85\% \\
  SRG IM & BBB+ & Inv & 1.73\% & 10.53\% & 1.56\% & 10.27\% & 0.02\% & 11.04\% \\
  DAI GY & BBB+ & Inv & 3.26\% & 20.94\% & 3.00\% & 21.06\% & 0.55\% & 34.13\% \\
  VIE FP & BBB & Inv & 2.62\% & 14.96\% & 2.39\% & 14.77\% & 0.21\% & 17.14\% \\
  \hline
  AMP IM & BB+ & Spec & 0.50\% & 4.82\% & 0.47\% & 4.66\% & 0.00\% &  3.16\% \\
  FR FP & BB+ & Spec & 3.12\% & 21.63\% & 2.90\% & 21.84\% & 0.62\% & 33.48\% \\
  EO FP & BB & Spec & 3.06\% & 22.02\% & 2.87\% & 22.13\% & 0.65\% & 29.00\%  \\
  GET FP & BB- & Spec & 1.50\% & 11.43\% & 1.38\% & 11.26\% & 0.03\% &  12.11\% \\
  LHA GY & BB- & Spec & 7.29\%  & 34.11\% & 6.93\% & 34.67\% & 5.09\% & 49.57\%  \\
  PIA IM & B+ & Spec & 1.08\% & 11.22\% & 1.02\% & 11.10\% & 0.02\% & 11.66\% \\
  CO FP & B & Spec & 5.62\% & 24.97\% & 5.29\% & 24.94\% & 2.48\% & 29.28\% \\
  \hline
\end{tabular}
\end{table}

\begin{itemize}
    \item[-] We can clearly observe that, both in the NegGamma and NegIG models, 1 year default probabilities are significantly higher than the Merton ones, except for very well rated issuers (A and beyond), which is coherent given the high level of quality of their signature (but, even for these issuers, NegGamma and NegIG probabilities remain strictly bigger than 0, which is not the case with the Merton model); it is also interesting to note that, for longer horizons, most NegGamma and NegIG default probabilities grow slower than Merton default probabilities, notably for speculative issuers. This risk redistribution towards shorter maturities is a clear improvement of classical diffusion models, which underestimate (resp. overestimate) short (resp. long) term default probabilities.
    
    \item[-] NegGamma probabilities are generally higher than NegIG probabilities: this is coherent with the fact that the L\'evy measure decreases in $1/|x|$ in the NegGamma model and in $1/|x|^{3/2}$ in the NegIG model, making for a fatter left tail in the NegGamma model and therefore for more frequent occurrences of downward jumps.
    
    \item[-] On the speculative segment, 1 year default probabilities are typically around 1-3\% in the NegIG and NegGamma models, which is again a clear improvement when compared to the Merton model, whose default probabilities lie around 0-0.5\%, except for very indebted firms with a $V_A/K$ ratio close to 1 such as LHA GY (which, as an airline company, has particularly suffered from the pandemics and saw its stock price $V_E$ drop by 35\% in 2020) or CO FP. It is also interesting to note that, for these 2 particular issuer, NegGamma and NegIG probabilities (5-7\%) remain significantly higher than the Merton probabilities (2-5\%).
    
    
    \item[-] On the same segment but for a time horizon T=5 years, the Merton model returns very high default probabilities (up to 50\% for LHA GY), while the NegIG model returns far more realistic figures (10-20\%, except for the two very endebted issuers cited above, but figures remain far lower than the Merton probabilities).


\end{itemize}

\begin{figure}[ht]
\centering     
\includegraphics[width=.8\textwidth]{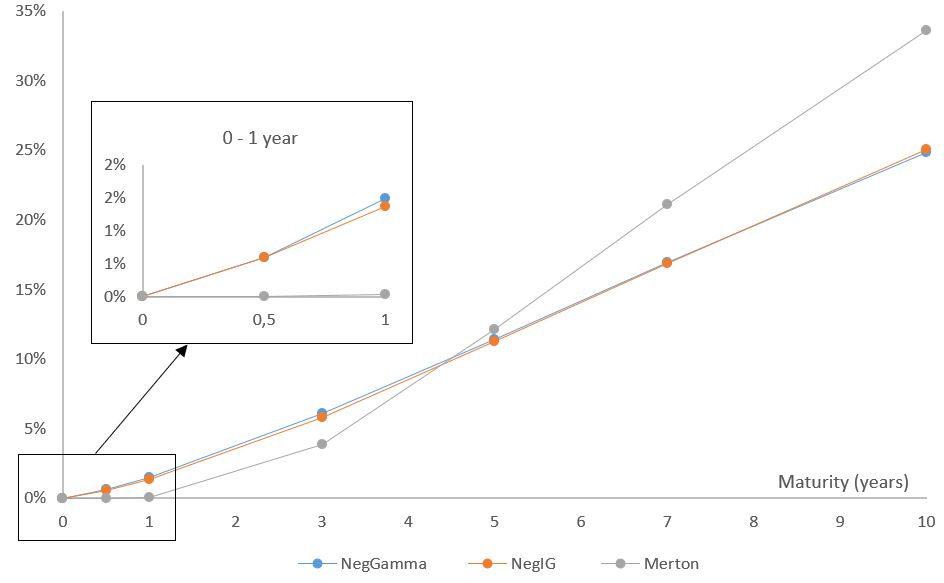}
\caption{Term structure for the risk-netural default probabilities given by several models, for the issuer GET FP. NegGamma and NegIG probabilities display a more regular growth than Merton ones, and are significantly higher for short horizons.} 
\label{fig:term_structure_FRFP}
\end{figure}

In Fig. \ref{fig:term_structure_FRFP}, we plot the term structure of the default probabilities for the GET FP issuer (0 to 10 years), in the NegGamma, NegIG and Merton models. We can recognize the characteristic shape of the Merton default probabilities: short term probabilities are very low, with a sudden increase and an inflection around 5 years and a flattening for longer horizons (which is a direct consequence of the shape of the normal cumulative distribution function N(.)). On the contrary, NegGamma and NegIG default probabilities display a more regular structure across maturities, and a clearly visible risk redistribution towards shorter time horizons, while Merton default probabilities are null for T$\leq$ 1, and reach 35\% on the 10 years horizon, which is completely incoherent with the historical default rates of this category.

\subsubsection{Actual (real-world) default probabilities}

In this subsection, we discuss actual (real world) default probabilities \eqref{def:actual_metrics}. Model parameters are calibrated using algorithm \ref{algorithm_calibration}, while the actual drift $\overline{r}$ is calibrated for each issuer thanks to formula \eqref{actual_drift}, in both cases using market data on a 2 year period. This is to avoid enormous drifts such as $\pm 50\%$ that can arise on a 1 year period (notably 2020) and lead to absurdly small or big default probabilities. In table \ref{tab:negGamma_negNIG_Merton_calibrated 2Y}, we compare the actual default probabilities obtained for the set of issuers (see details in appendix \ref{app:data_set}) with our NegGamma and NegIG models to the one obtained via the classical Merton model, for two distinct horizons $T$ (1 year and 5 years). We also compare with historical default rates of the investment grade and speculative categories displayed in table \ref{tab:histo_rate}.

\begin{table}[ht]
 \caption{Historical default rates for the investment grade and speculative categories (source: \cite{SP20b}).}
 \label{tab:histo_rate}       
 \centering
 \begin{tabular}{c|cc|cc|cc}
 \hline
  \multirow{ 2}{*}{Category} & \multicolumn{2}{c|}{1 year default rate} & \multicolumn{2}{c|}{5 years default rate} &
  \multicolumn{2}{c}{10 years default rate} \\
   & 2010-2019 & 2009 & 2010-2015 & 2009 & 2000-2008 & 2009\\
   \hline
 Investment Grade & 0.01\% & 0.33\% & 0.17\% & 0.63\% & 2.15\% & 0.92\% \\
 Speculative & 2.53\% & 9.95\% & 10.01\% & 16.47\% & 21.3\% & 20.74\% \\
 All ratings & 1.19\% & 4.19\% & 4.61\% & 6.99\% & 9.43\% & 8.88\% \\
 \hline
\end{tabular}
\end{table}

\begin{table}[ht]
 \caption{Comparison of NegGamma, NegIG and Merton actual default probabilities, for time horizons T=1 and 5 years and for several categories/ratings.}
 \label{tab:negGamma_negNIG_Merton_calibrated 2Y}       
 \centering
 \begin{tabular}{lcc|cc|cc|cc}
 \hline
  & Rating & Category &  \multicolumn{2}{c|}{NegGamma} & 
  \multicolumn{2}{c|}{NegIG} & 
  \multicolumn{2}{c}{Merton}
  \\ 
  \hline
  & &  & $T$= 1 y. & $T$=5 y. & $T$= 1 y. & $T$=5 y. &  $T$=1 y. & $T$=5 y. \\
  SAP GY & A & Inv & 0.01\% & 0.04\% & 0.01\% &  0.04\% & 0.00\% &  0.00\% \\
  MRK GY & A & Inv & 0.02\% & 0.09\% & 0.03\%  & 0.09\% & 0.00\% &  0.00\% \\
  AI FP & A- & Inv & 0.02\% & 0.05\%  & 0.02\% & 0.05\% & 0.00\% &  0.00\% \\
  SU FP & A- & Inv & 0.03\%  & 0.06\% & 0.03\% & 0.06\% & 0.00\% &  0.00\% \\
  CRH LN & BBB+ & Inv & 0.37\% & 1.56\% & 0.35\% & 1.50\% & 0.00\% & 0.22\%  \\
  SRG IM & BBB+ & Inv & 0.83\% & 3.39\% & 0.74\% & 3.18\% & 0.00\% & 0.82\% \\
  DAI GY & BBB+ & Inv & 2.11\% & 13.72\% &  1.92\% & 13.50\% & 0.10\%  & 15.76\%  \\
  VIE FP & BBB & Inv &  1.39\% & 7.52\% & 1.25\% & 7.24\% & 0.01\% & 5.16\%  \\
    \hline
  AMP IM & BB+ & Spec & 0.12\% &  0.16\% & 0.12\% & 0.16\% & 0.00\%  &  0.00\% \\
  FR FP & BB+ & Spec & 1.93\% & 11.07\% & 1.79\%  & 10.91\% & 0.12\% & 11.25\%  \\
  EO FP & BB & Spec & 2.72\%  & 20.02\% & 2.57\% & 19.98\% & 0.52\% & 22.42\%   \\
  GET FP & BB- & Spec & 0.68\% & 4.00\% & 0.63\% & 3.83\% & 0.00\% & 1.64\% \\
  LHA GY & BB- & Spec & 9.45\% & 81.73\% & 9.13\% & 81.93\% & 8.56\% & 82.53\%  \\
  PIA IM & B+ & Spec & 0.53\% & 1.71\% & 0.41\% & 1.66\% & 0.00\%  & 0.46\%  \\
  CO FP & B & Spec & 13.82\% & 88.99\% & 13.57\% & 88.82\% & 14.67\% & 87.14\%    \\
  \hline
\end{tabular}
\end{table}

We can observe that:
\begin{itemize}
    \item[-] Like in the risk-neutral case, actual default probabilities are significantly higher in the NegGamma and NegIG models than in the Merton model;
    
    \item[-] On the speculative segment, 1 year actual default probabilities are typically around 1-3\% in the NegIG and NegGamma models (except for the very indebted issuers already discussed), which provides a very good agreement with the 2010-2019 average for 1 year default rates in this category (2.53\%).
    
    \item[-] Still on the speculative segment but for longer (5Y) maturities, actual default probabilities of all models tend to converge to similar values, highlighting the predominant role of the drift for longer horizons, except for issuers with the best signature (such as AMP or PIA) where NegGamma or NegIG default probabilities remain higher than the Merton ones, and closer to 5Y historical default rates.
    
    \item[-] It is interesting to note that, for the best rated issuers retained in our study, 1 and 5 years actual default probabilities provide a remarkably good fit to historical default rates, while Merton probabilities remain equal to $0\%$ for all maturities. For the other issuers of the investment grade segment, NegGamma and NegIG default probabilities appear to slightly overestimate short term historical default rates; we may note, however, that these probabilities can easily be diminished by usual data re-processings of the equity or debt values for such issuers. For instance, an important part of a firm's debt can be concentrated by one of its financial captive, as frequently the case in the automobile industry (DAI GY); it is therefore common to lower the face value of the total debt $K$ by the captive debt amount. Other issuers, for instance operating in the energy sector (like VIE FP), require expensive equipment with long term financing, and in that case a common practice is to integrate operational cashflows to the market capitalization ($V_E$) of the firm when evaluating its default risk.
\end{itemize}

\subsection{Convergence analysis}

In table \ref{tab:conv_analysis}, we indicate the minimal, maximal and average number of steps needed to reach a precision of $10^{-3}$ when applying algorithm \ref{algorithm_calibration} to the data set of table \ref{tab:data_set}, as well as the corresponding elapsed time (in seconds). The calculations are made on a personal computer with Intel(R) Core(TM) i5-8265U CPU @1.60GHz. Clearly, one can observe that the number of steps, although increasing with maturity, remains very small and that elapsed time is around or even less than a second.

\begin{table}[ht]
 \caption{Comparison the convergence speed of algorithm \ref{algorithm_calibration} for NegGamma, NegIG and Merton models for various time horizons.}
 \label{tab:conv_analysis}       
 \centering
 \begin{tabular}{c|ccc|ccc|ccc}
 \hline
  &  \multicolumn{3}{c|}{NegGamma} & 
  \multicolumn{3}{c|}{NegIG} & 
  \multicolumn{3}{c}{Merton}
  \\ 
  \hline
  Number of steps & Min & Max & Av. & Min & Max & Av. & Min & Max & Av. \\
   \hline
      $T =$ 1Y & 3 & 7 & 4.48 & 2 & 6 & 3.71 & 2 & 7 & 4.05 \\
      $T =$ 5Y & 3 & 11 & 5.95 & 3 & 8 & 4.76 & 3 & 13 & 6.61 \\
      $T =$ 10Y & 3 & 13 & 6.90 & 3 & 10 & 5.38 & 3 & 17 & 8.24 \\
      $T =$ 15Y & 4 & 15 & 7.43 & 3 & 11 & 6.00 & 4 & 18 & 8.86 \\
    \hline
    Elapsed time (in sec) & Min & Max & Av. & Min & Max & Av. & Min & Max & Av. \\
    \hline
      $T =$ 1Y & 0.48 & 1.25 & 0.74 & 0.16 & 0.53 & 0.32 & 0.08 & 0.31 & 0.16 \\
      $T =$ 5Y & 0.57 & 2.23 & 1.11 & 0.25 & 0.71 & 0.42 & 0.11 & 0.55 & 0.27 \\
      $T =$ 10Y & 0.56 & 2.65 & 1.33 & 0.29 & 0.88 & 0.49 & 0.12 & 0.61 & 0.31 \\
      $T =$ 15Y & 0.75 & 3.30 & 1.45 & 0.33 & 1.12 & 0.58 & 0.15 & 0.63 & 0.33 \\
  \hline
\end{tabular}
\end{table}

\section{A symmetric model}\label{sec:symmetric}

Let us now consider the case where the stochastic process $X_t$ in \eqref{SDE} is a symmetric L\'evy process, that is, a process with equally probable downward and upward jumps. We will consider a Variance Gamma (VG) process, which, for simplicity, will be assumed to be symmetric. VG processes have already been introduced in credit risk models in \cite{Fiorani06,Fiorani10}, but using a calibration from CDS indices (Markit CDX.NA.HY and CDX.NA.IG). Again, we adopt a different approach, by taking advantage of recently derived pricing formulas for the European call to perform algorithm \ref{algorithm_calibration} on each firm's own equity, and deduce its default probability.

\subsection{The symmetric Variance Gamma (symVG) credit risk model}\label{subsec:sym_VG}

The Variance Gamma (VG) process, which has been popularized in financial modelling in \cite{Madan98}, can be defined by a time changed drifted Brownian motion:
\begin{equation}\label{VG_process}
    V_t \, := \, \theta G_t \, + \, \sigma  W_{G_t},
\end{equation}
where $G_t$ is a Gamma process of mean $\mu:=1$ and variance $\nu > 0$ (note that this parametrization is equivalent to the rate/shape parametrization of subsection \ref{subsec:Neg_Gamma}, for $\rho=\mu^2 / \nu$ and $\lambda= \mu / \nu$). $\sigma > 0$ is the scale parameter and $\theta\in\mathbb{R}$ the asymmetry, or sknewness parameter; in the case $\theta = 0$ the distribution of the VG process is symmetric around the origin (this model was introduced earlier, in \cite{Madan90}). The definition \eqref{VG_process} shows that the VG process is actually a so-called normal mean mixture \citep{Barndorff82}, where the mixing distribution is given by the Gamma distribution and materializes the passage of the business time, jumps in the process materializing periods of intense trading activity. When $\nu$ grows high then jumps are more frequent and the business time admits staircase-like realizations; on the contrary, the $\nu\rightarrow 0$ configuration corresponds to a linear passage of time.

The L\'evy measure of the VG process is known to be
\begin{equation}\label{VG_measure}
    \Pi_{V}( \ud x) \, = \, \frac{e^{\frac{\theta x}{\sigma^2}}}{\nu |x|} e^{ -\frac{\sqrt{  \frac{\theta^2}{\sigma^2} + \frac{2}{\nu} }}{\sigma} |x| } \, \ud x
\end{equation}
which shows that the VG process is a particular case of a CGMY process \citep{Carr02}. As the measure \eqref{VG_measure} is defined on the whole real axis, realizations of the VG process feature both upward and downward jumps; while downward jumps still materialize a brutal drop in a firm's assets, we may think of upward jumps as cash injections from governmental institutions or central banks in order to help recapizalizing severely indebted issuers. Let us also remark that when $\theta=0$ then the L\'evy measure is symmetric around 0 (and in that case positive and negative jumps occur with equal probability). The L\'evy symbol can be computed from \eqref{VG_measure} and the L\'evy-Khintchine formula \eqref{Levy_Khintchine}, resulting in
\begin{equation}\label{psi_VG}
    \psi_V(u)
    \, = \, 
    - \frac{1}{\nu} 
    \log\left( 1 - i\theta\nu u + \frac{\sigma^2\nu}{2} u^2 \right)
\end{equation}
from which we easily deduce the cumulant generating function as well as the first cumulants and central moments:
\begin{equation}\label{cumulants_symVG}
    \left\{
    \begin{aligned}
        &  \overline{\mu}_{V}^{(1)} \,  = \,  \overline{r} + \frac{1}{\nu}\log \left( 1 - \theta\nu - \frac{\sigma^2\nu}{2}  \right) + \theta \quad \mathrm{(Mean)}  \\
        & \mu_V^{(2)} \,  = \, \theta^2\nu + \sigma^2 \quad \mathrm{(Variance)} \\
        & \mu_V^{(3)}   \, = \, \frac{\theta\nu(2\theta^2\nu+3\sigma^2)}{(\theta^2\nu + \sigma^2)^{3/2}}
        \quad \mathrm{(Skewness)} \\
        & \mu_V^{(4)} \, = \,
        \frac{
        3\sigma^4\nu + 12 \sigma^2\theta^2\nu^2 + 6 \theta^4\nu^3 + 3\sigma^4 + 6\sigma^2\theta^2\nu + 3\theta^4\nu^2
        }{(\theta^2\nu + \sigma^2)^2}
        \quad \mathrm{(Kurtosis)}
        .
    \end{aligned}
    \right.
\end{equation}
If we assume $\theta=0$, the tails are symmetric and the skewness is equal to $0$ and the kurtosis becomes
\begin{equation}
    \mu_V^{(4)} \overset{\theta\rightarrow 0}{\longrightarrow} 3(\nu+1)
    ,
\end{equation}
highlighting the role of $\nu$ as the measure of excess kurtosis of the distribution.

\subsection{Default metrics}\label{subsec:sym_VG_metrics}
When the VG process $V_t$ is symmetric, its density is known to be
    \begin{equation}
        f_{V}(x,t) \, = \, \frac{2}{ \sqrt{2\pi} \Gamma(\frac{t}{\nu}) (\sigma^2\nu)^{\frac{t}{\nu}} } \, 
        \left( \frac{ |x| }{ \frac{1}{\sigma} \sqrt{\frac{2}{\nu}} }  \right)^{\frac{t}{\nu}-\frac{1}{2}} \, 
        K_{\frac{t}{\nu}-\frac{1}{2}} \left( \frac{1}{\sigma} \sqrt{\frac{2}{\nu}} |x| \right)
    \end{equation}
and, taking $\theta=0$ in \eqref{psi_VG}, we easily obtain the convexity adjustment
\begin{equation}\label{omega_VG}
    \omega_V \, = \, \frac{1}{\nu}\log \left( 1 - \frac{\sigma^2\nu}{2}  \right)
    .
\end{equation}
From definition \ref{def:metrics}, the distance to default in the symVG model is therefore
\begin{equation}
    k_V \, = \, \log\frac{V_A}{K}
    + \left( r + \frac{1}{\nu}\log \left( 1 - \frac{\sigma^2\nu}{2}  \right)\right) T
\end{equation}

\begin{remark}
    In the low kurtosis regime $\nu\rightarrow 0$, Taylor expanding \eqref{omega_VG} yields $\omega_V \sim -\sigma^2 /2$, thus recovering the Gaussian adjustment; in other words, the symVG model recovers the Merton model in its low kurtosis limit.
\end{remark}
The cumulative distribution function of the Variance Gamma distribution is not known in exact form but easily accessible from any programming language. For instance in R one can use the CDF function applied to the Variance Gamma density (denoted by "dvg" in the Variance Gamma package), and the default probability in the symVG model can be written formally as:
\begin{equation}
    F_V(-k_V) \, = \, 
    \mathrm{CDF}\left( \mathrm{dvg} 
    \left( -\frac{k_V}{\sqrt{T}}, \, \mu=0, \, \frac{\sigma}{\sqrt{T}}, \, \theta = 0  , \, \frac{\nu}{t} \right)  \right)
    .
\end{equation}

\subsection{Equity value}\label{subsec:sym_VG_equity}
An exact formula for the price of the European call under the VG process has been established in \cite{Madan98}, but it involves complicated products of Bessel and hypergeometric functions; in practice, Fourier techniques are favored \citep{Lewis01}, notably because of the relative simplicity of the characteristic exponent \eqref{psi_VG}. This technology, however, is not adapted to the implementation of algorithm \ref{algorithm_calibration} because if would call for solving an equation involving an integral in the Fourier space. Instead, we propose to use a recently derived pricing formula for the European call in the symmetric VG model, which takes the form of quickly convergent series of powers of the log forward moneyness (or, in the context of default modelling, of the distance to default). To that extent, let us introduce the notations
\begin{equation}
    T_\nu \, := \, \frac{T}{\nu} - \frac{1}{2}
    \hspace{1cm}
    \sigma_\nu \, := \, \sigma\sqrt{\frac{\nu}{2}}
    ,
\end{equation}
and assume $T_\nu\notin\mathbb{Q}$. For $X,Y\in\mathbb{R}$, let us also define the quantity

\begin{multline}
    a_{n_1,n_2}(X,Y) \, := \,
    \frac{(-1)^{n_1}}{n_1!}
    \left[
     \frac{\Gamma(\frac{-n_1+n_2+1}{2} + T_\nu)}{\Gamma(\frac{-n_1+n_2}{2}+1)} 
     \left(-\frac{X}{Y}\right)^{n_1} Y^{n_2} \right.
     \\
    \left.
    +
    \, 2  \, \frac{\Gamma(-2n_1-n_2-1 -2T_\nu)}{\Gamma(-n_1+\frac{1}{2}-T_\nu)} 
     \left(-\frac{X}{Y}\right)^{2n_1+1+2T_\nu} (-X)^{n_2}
    \right]
    .
\end{multline}

\begin{proposition}\label{prop:symVG_equity_price}
    The equity value of a firm in the symVG credit risk model can be written as:
     \begin{itemize}
        \item[-]
         If $ k_{V} < 0$, 
        \begin{equation}
            V_E \, = \, \frac{Ke^{-rT}}{2\Gamma(\frac{T}{\nu})} \sum\limits_{\substack{n_1 = 0 \\ n_2 = 1}}^\infty
            a_{n_1,n_2}(k_V,\sigma_\nu)
            .
        \end{equation}
       \item[-]
         If $ k_{V} > 0$, 
        \begin{equation}
            V_E \, = \, 
            V_A - K e^{-rT} - \frac{Ke^{-rT}}{2\Gamma(\frac{T}{\nu})} \sum\limits_{\substack{n_1 = 0 \\ n_2 = 1}}^\infty
             a_{n_1,n_2}(k_V,-\sigma_\nu)
             .
        \end{equation}
        \item[-]
        If $ k_{VG} = 0 $,
        \begin{equation}                   V_E \, = \,                  \frac{Ke^{-rT}}{2\Gamma(\frac{T}{\nu})} \sum\limits_{n=1}^\infty \frac{\Gamma(\frac{n+1}{2}+T_\nu)}{\Gamma(\frac{n}{2}+1)} \, \sigma_\nu^n 
            .
        \end{equation} 
        \end{itemize}
\end{proposition}
\begin{proof}
    See \cite{Aguilar20a}.
\end{proof}

\subsection{Parameter calibration and default probabilities}\label{subsec:sym_VG_calibration}

The calibration procedure is, like for one-sided models, based on algorithm \ref{algorithm_calibration} and proposition \ref{prop:symVG_equity_price}. Note that, typically, it is sufficient to truncate the series at $n=m=10$ to get a precision of $10^{-4}$ on equity prices, and that machine precision is reached between 15 and 20 iterations (see an example on the issuer GET FP in table \ref{tab:sym_VG_convergence}); convergence can be slightly slower if the ratio $V_A/K$ becomes very small, but remains very efficient in all cases (and is particularly accelerated for short maturities). We also refer to \cite{AK21} for a precise discussion and comparison of the pricing formula in proposition \ref{prop:symVG_equity_price} with state of the art numerical pricing techniques such as Fast Fourier Transform and their recent refinements (COS, PROJ etc.).

Given the moments \eqref{cumulants_symVG} for $\theta=0$, we need to consider the function
\begin{equation}\label{M_V}
    \M_V \, : 
    \begin{bmatrix}
        \sigma \\
        \nu
    \end{bmatrix}
    \, \longrightarrow \, 
    \begin{bmatrix}
        \sigma^2 \\
        3\nu
    \end{bmatrix}
\end{equation}
whose inversion is immediate. Calibration results and default probabilities are displayed in table \ref{tab:sym_VG}; we only estimate risk-neutral default probabilities; estimation of real-world is straightforward using the same technique as for one-sided models described in subsection \ref{subsec:actual_drift}. 

\begin{table}[ht]
 \caption{Parameter calibration and 1 year risk-neutral default probabilities in the symVG model, compared with NegGamma and Merton default probabilities}
 \label{tab:sym_VG}       
 \centering
 \begin{tabular}{l|ccc|ccc}
  \hline
  \multirow{ 2}{*}{Issuer} & \multicolumn{3}{|c|}{Calibrated parameters} & \multicolumn{3}{c}{Default probabilities} \\
  & $\sigma$ & $\nu$ & $V_A$ (in MM) & SymVG & Merton & NegGamma  \\
  \hline
  SAP GY & 0.2873 & 2.2526 & 180 904 & 0.00 \% & 0.00\% & 0.01\% \\
  AI FP & 0.2340 & 3.5774 & 78 923 & 0.03 \% & 0.00\% & 0.08\% \\
  SU FP & 0.3246 & 3.9171 & 71 472.2 & 0.07\% & 0.00\% & 0.14\% \\
  CRH LN & 0.3038 & 2.9005 & 33 948.2 & 0.57\% & 0.01\% & 1.10\%  \\
  SRG IM & 0.1942 & 6.3719 & 29 556.2 & 0.91\% &0.02\% & 1.73\%  \\
  DAI GY & 0.1026 & 3.5876  & 213800  & 1.61\% & 0.55\% & 3.26\%  \\
  VIE FP & 0.1618 & 4.4099 & 27 281.3 & 1.42\% & 0.21\% & 2.62\% \\
  \hline
  AMP IM & 0.3598 & 4.809 & 8 627.48 & 0.26\% & 0.00\% & 0.50\% \\
  FR FP & 0.2741 & 3.1256 & 11 931.5 & 1.02\% & 0.62\% & 3.12\% \\
  EO FP & 0.2754 & 1.7378 & 10 010.3  & 1.69\% & 0.65\% & 3.06\% \\
  GET FP & 0.2402 & 3.2453 & 11 666.7 & 0.79\% & 0.03\% & 1.50\% \\
  LHA GY & 0.2092  & 2.5558 & 14 700 & 4.40\% & 5.09\% & 7.29\% \\
  PIA IM & 0.2513 & 1.8526 & 1 494.92 & 0.57\% & 0.02\% & 1.08\%  \\
  CO FP & 0.0713 & 2.6652 & 16 427.4 & 3.30\% & 2.48\% & 5.62\% \\
 \hline
\end{tabular}
\end{table}

It is interesting to note that 1 year symVG default probabilities are all comprised between 1 year Merton and 1 year NegGamma default probabilities (except for LHA GY). The fact that symVG probabilities are lower than the probabilities computed in a one-sided model (the NegGamma model) is no surprise, since the symVG model features also positive jumps (i.e. sudden raises in asset prices), which can diminish the probability for $V_A(T)$ to be lower than $K$; the symVG model is therefore more tempered than the NegGamma model, but its default probabilities remain very significantly higher than in the Merton model, except for top quality issuers.

\begin{table}[ht]
 \caption{Convergence of the equity formula in proposition \ref{prop:symVG_equity_price} for a test issuer (GET FP) and for various truncations at $n=m=N_{max}$. Parameters (estimated in table \ref{tab:sym_VG}): $V_A=$11 666.7, $K=$ 4 998, $T=1$, $r=0$, $\sigma=0.2402$, $\nu=3.2453$.}
 \label{tab:sym_VG_convergence}       
 \centering
 \begin{tabular}{ccccc}
  \hline
  & $N_{max}=7$ & $N_{max}=10$ & $N_{max}=15$ & $N_{max}=20$  \\
  \hline
  $V_E$  & 6 693.990  & 6 676.968  & 6 676.847  & 6 676.847 \\
  Relative error & 0.256760\% & 0.001815\%  & 0.000001\%   & 0\% (machine precision) \\ 
 \hline
\end{tabular}
\end{table}




\section{Concluding remarks}\label{sec:conclusion}

In this study, we have generalized the notions of distance to default and default probabilities to the class of pure jump structural models, and studied two sub-classes: one-sided models (NegGamma, NegIG) and symmetric models (symVG). One-sided models are particularly tractable, thanks to the closed formulas for the equity value that we have established (propositions \ref{prop:GammaNeg_equity_value} and \ref{prop:IGNeg_equity_value}). Such closed formulas make the calibration algorithm easy to implement.

One-sided models produce higher short term default probabilities (both risk-neutral and real world) than the Merton model. In particular, these probabilities typically lie between 0.5\% and 3\% for speculative issuers, which is close to the historical 1 year default rate of this category. On the investment grade segment, high quality issuers (with a ranking higher or equal to A) have a small default probability (0.01-0.15\%) which is coherent with the quality of their signature; the figures, however, remain clearly higher than the Merton ones.

In double-sided models (symVG), the presence of upward jumps moderates the probability for $V_A(T)$ to be lower than $K$, and therefore default probabilities are lower than NegGamma and NegIG default probabilities, while remaining higher than the Merton ones. The model is, consequently, well suited to investment grade issuers for longer horizons, but as the equity value is only accessible via truncated series expansions, the calibration algorithm takes more time to converge.

Future work should include the consideration of other dynamics for asset prices, such as a double-sided process with asymmetric distribution of jumps; recent pricing formulas established for instance in the case of an asymmetric NIG process could be used for computing the equity value, but their series nature will call for several optimization in the calibration algorithm to maintain a satisfactory level of performance. 

\section*{Acknowledgments}
We thank Alexander Lipton for his interest in this work and his many valuable comments. We also thank two anonymous reviewers whose detailed reports greatly helped us improve the quality of the manuscript.

\singlespacing

\doublespacing

\appendix

\section{Notations used in the paper}\label{app:notations}

\paragraph{Special functions}
\hspace{0.5cm}

\noindent Standard normal cumulative distribution function:
\begin{equation}
    N(x) \, := \, 
    \frac{1}{\sqrt{2\pi}} \,
    \int\limits_{-\infty}^{x} e^{-\frac{y^2}{2}} \, \ud y
\end{equation}

\noindent Lower incomplete Gamma function:
\begin{equation}\label{lower_gamma}
    \gamma(a,z) \, := \, \int\limits_0^{z} e^{-x} x^{a-1} \, \ud x
    .
\end{equation}

\noindent Upper incomplete Gamma function:
\begin{equation}\label{upper_gamma}
    \Gamma(a,z) \, := \, \int\limits_z^{\infty} e^{-x} x^{a-1} \, \ud x
    .
\end{equation}

\noindent Bessel function of the second kind (MacDonald function):
\begin{equation}
    \K_\nu(z) \, := \, \frac{1}{2} \, \left(\frac{z}{2}\right)^\nu \,  \int\limits_0^{\infty} \, 
    e^{-x - \frac{z^2}{4x} } \, x^{-\nu - 1} \, \ud x
    .
\end{equation}

\noindent Shuster's integrals
\footnote{This is a consequence of the Cauchy - Schl\"{o}milch substitution
\begin{equation}
    \int\limits_\mathbb{R} f(x) \, \ud x \, = \, \int\limits_\mathbb{R} f\left( x - \frac{1}{x} \right) \, \ud x
\end{equation}
which is itself a consequence of the more general Glasser's master theorem \citep{Glasser83}.
} \citep{Shuster68}:
\begin{equation}\label{Shuster}
    \left\{
    \begin{aligned}
        & \int\limits_{0}^x \frac{e^{-\lambda\frac{(y-\mu t)^2}{2\mu^2 y}}}{y^{\frac{3}{2}}} \, \ud y \, = \, \sqrt{\frac{2\pi}{\lambda t^2}}
        \left[ 
        N \left( \sqrt{\frac{\lambda t^2}{x}} \left( \frac{x}{\mu t} -1 \right) \right)
        \, + \, 
        e^{2\frac{\lambda t}{\mu}} \, 
        N \left( - \sqrt{\frac{\lambda t^2}{x}} \left( \frac{x}{\mu t} + 1 \right) \right)
        \right]
        \\
        & \int\limits_{x}^\infty \frac{e^{-\lambda\frac{(y-\mu t)^2}{2\mu^2 y}}}{y^{\frac{3}{2}}} \, \ud y \, = \, 
        \sqrt{\frac{2\pi}{\lambda t^2}}
        \left[ 
        N \left( - \sqrt{\frac{\lambda t^2}{x}} \left( \frac{x}{\mu t} -1 \right) \right)
        \, - \, 
        e^{2\frac{\lambda t}{\mu}} \, 
        N \left( - \sqrt{\frac{\lambda t^2}{x}} \left( \frac{x}{\mu t} + 1 \right) \right)
        \right]
        .
    \end{aligned}
    \right.
\end{equation}

\paragraph{Model notations}
\hspace{0.5cm}

\noindent Characteristic function and L\'evy symbol:
\begin{equation}
    \Psi_X(u,t) \, := \, \mathbb{E} [e^{i u X_t}] \, := \, e^{t\psi_X(u)}
    .
\end{equation}

\noindent Moment generating function (double-sided Laplace transform):
\begin{equation}
    M_X(p,t) \, := \, \mathbb{E} [e^{p X_t}] \, = \, \Psi(-ip,t)
    .
\end{equation}

\noindent Cumulant generating function:
\begin{equation}
    \kappa_X(p) \, := \, \log M_X(p,1) \, = \,  \psi_X(-ip)
    .
\end{equation}

\noindent $n^{th}$ cumulant:
\begin{equation}
    \kappa_X^{(n)} \, := \, \frac{\partial^n \kappa_X(p)}{\partial p^n}\big\vert_{p=0}
    .
\end{equation}

\noindent Standardized central moments of a time series $\mathcal{T}_X$:
\begin{equation}
    \mu_X^{(1)} \, := \, \mathbb{E} \left[ \T_X \right]
    , \quad
    \mu_X^{(2)} \, := \, \mathbb{E} \left[ (\T_X-\mu_X^{(1)})^2 \right]
    ,
\end{equation}
\begin{equation}\label{def_moments}
    \mu_X^{(n)} \, := \, 
    \frac{\mathbb{E} \left[ (\T_X-\mu_X^{(1)})^k \right]}
    {\left( \mu_X^{(2)}\right)^{k/2} }
   \quad
   (n > 2)
\end{equation}

\noindent Subset of standardized central moments of order $n \geq 2$:
\begin{equation}\label{def_moments_set}
    \mu_X^+ \, := \, \left\{  \mu_X^{(n)} , n\geq 2 \right\}
    .
\end{equation}

\noindent Cumulative distribution function (process with density $f_X(x,t)$):
\begin{equation}
    F_X(x) \, := \, \int\limits_{-\infty}^{x} \, f_X(y,T) \, \ud y
    .
\end{equation}

\noindent Martingale adjustment:
\begin{equation}
    \omega_X \, := \, -\psi_X(-i) \, = \, -\kappa_X(1) = \ -\log\mathbb{E} \left[ e^{X_1} \right]
    .
\end{equation}

\noindent Risk-neutral and Actual distance to default:
\begin{equation}
    k_X \, := \, \log\frac{V_A}{K} + (r+\omega_X)T
    ,
    \quad
    \overline{k}_X \, := \, \log\frac{V_A}{K} +  (\overline{r}+\omega_X)T
    .
\end{equation}

\noindent Risk-neutral and Actual default probability:
\begin{equation}
    P_X \, := \, F_X(-k_X)
    ,
    \quad
    \overline{P}_X \, := \, F_X(-\overline{k}_X)
    .
\end{equation}

\section{Data set}\label{app:data_set}
To test the models, we have chosen a set of issuers from various European markets, which are representative of several industrial sectors and S\&P rating categories. In table \ref{tab:data_set} we provide details on the issuers and the debt level chosen in our study.

\begin{table}[ht]
 \caption{Data set used in the paper; debt values are as of 13 oct. 2020. Source: Bloomberg.}
 \label{tab:data_set}       
 \centering
 \begin{tabular}{ccccc}
 \hline
  Ticker & Name & Industry group & Total debt (in MM) & Category  \\  
  \hline
  SAP GY & SAP SE & Software & 16 196 & Investment \\
  MRK GY & Merck KGAA  & Pharmaceuticals & 14 180 & Investment \\
  AI FP & Air Liquide SA & Chemicals & 14 730 & Investment \\
  SU FP & Schneider Electric SE & Electrical Components 
  & 8 473 & Investment \\
  CRH LN & CRH PLC & Building materials & 10 525 & Investment \\
  SRG IM & Snam SPA & Gas & 14 774 & Investment \\
  DAI GY & Daimler AG & Auto manufacturers & 161 780 & Investment  \\
  VIE FP & Veolia Environment & Water & 16 996 & Investment  \\ 
  AMP IP & Amplifon SPA & Pharmaceuticals & 1 339 & Speculative \\
  FR FP & Valeo SA & Auto parts \& Equipment & 4 879 & Speculative\\
  EO FP & Faurecia & Auto parts \& Equipment & 4 838 & Speculative \\
  GET FP & Getlink SE & Transportation & 4 998 & Speculative \\
  LHA GY & Deutsche Lufthansa & Airlines & 10 106 & Speculative \\
  PIA IM & Piaggio & Motorcycles & 609 & Speculative \\
  CO FP & Casino & Foods & 14 308 & Speculative \\
  \hline
\end{tabular}
\end{table}

\section{Stability of the parameters}\label{app:time_stability}

In figure \ref{fig:ParamEvol}, we plot the evolution of calibrated parameters in function of debt maturity $T$, for 2 representative issuers and the NegGamma, NegIG and Merton models. The parameters do not significantly evolve with $T$; this is perfectly consistent with the fact that the time series $\T_X$ of the log-returns of $V_A(t)$ used for the calibration is not affected by the choice made for the debt maturity $T$.

\begin{figure}[h!t!b!]
\centering     
\subfigure{\includegraphics[width=.512\textwidth]{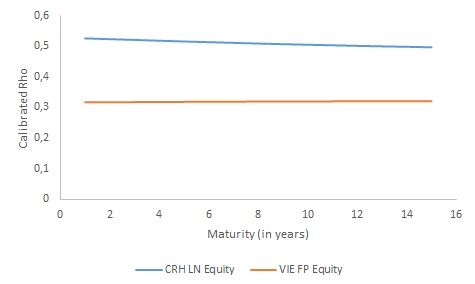}}\hspace{-1.9em}
\subfigure{\includegraphics[width=.512\textwidth]{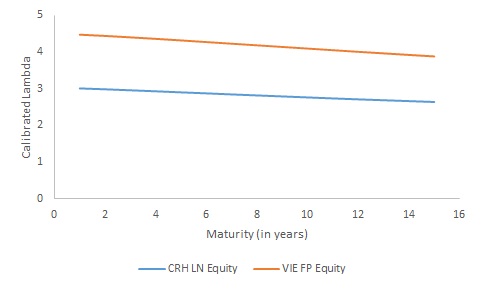}}
\subfigure{\includegraphics[width=.512\textwidth]{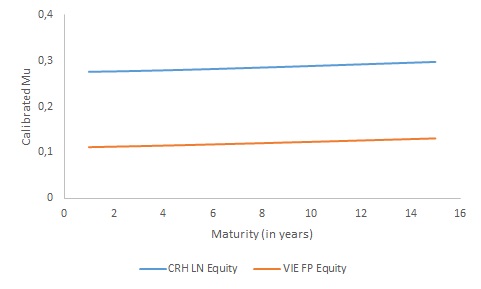}}\hspace{-1.9em}
\subfigure{\includegraphics[width=.512\textwidth]{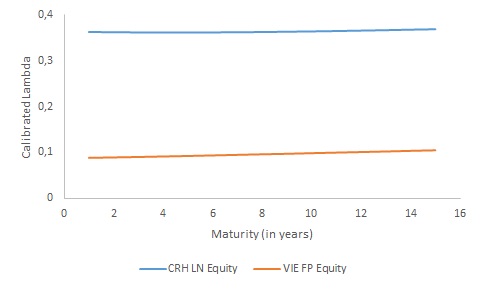}}
\subfigure{\includegraphics[width=.512\textwidth]{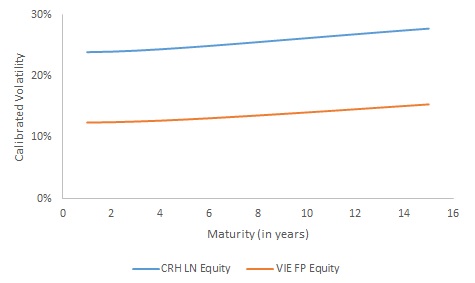}}
\caption{\small{Value of the parameters of each model for different maturities : the first line corresponds to the negGamma model (Left : $\rho$, Right : $\lambda$). The second line corresponds to the negIG model (Left : $\mu$, Right : $\lambda$). The third line corresponds to the value of $\sigma$ computed with the Merton model.
}}\label{fig:ParamEvol}
\end{figure}

\section{Code}\label{app:code}

The R code written to implement algorithm \ref{algorithm_calibration} on the whole data set displayed in table \ref{tab:data_set} and to estimate default probabilities is publicly available at \url{https://github.com/Voxinat/CreditRiskModels}.

\end{document}